\documentclass{ifacconf}

\usepackage{graphicx}      % include this line if your document contains figures
\usepackage{natbib}        % required for bibliography
\usepackage{amsfonts}
\usepackage{amsmath}
\usepackage{xcolor}
\usepackage{algorithm}
\usepackage{algpseudocode}
\usepackage{multirow}
\usepackage{makecell}
\usepackage{subcaption}

\newcommand{\Rpos}{\mathbb{R}_{\geq 0}}
\newcommand{\R}{\mathbb{R}}

\newcommand{\argmax}[1]{\underset{#1}{\textrm{argmax}}}
\newcommand{\argmin}[1]{\underset{#1}{\textrm{argmin}}}

\makeatletter
\let\theoremstyle\undefined
\makeatother

\usepackage{amsthm}
\newtheorem{theorem}{Theorem}
\newtheorem{lemma}{Lemma}

\newtheorem{corollary}{Corollary}

\DeclareMathOperator{\diag}{diag}
\begin{document}
\begin{frontmatter}

\title{Radioactive Source Seeking using Bayesian Optimisation with Movement Penalty}

\thanks[footnoteinfo]{This work was supported by an Australian Government Research Training Program (RTP) Scholarship, the UoM Ingenium Program, and DMTC Ltd. [Project: 15.188 UAV-mounted Scintillator Sensor for Standoff Detection of Radiological and Nuclear (RN) materials]. The authors have prepared this paper in accordance with the intellectual property rights granted to partners from the original DMTC project.}

\author[Mech]{Lysander Miller} 
\author[Mech]{Joshua Keene}
\author[OS]{Jeremy M.\ C.\ Brown}
\author[Mech]{Airlie Chapman}
% \author[Third]{Third C. Author}

\address[Mech]{Department of Mechanical Engineering, The University of Melbourne, Parkville, VIC 3010 Australia (e-mail: lysanderm@student.unimelb.edu.au; joshua.keene@student.unimelb.edu.au; airlie.chapman@unimelb.edu.au)}
\address[OS]{Optical Sciences Centre, Department of Physics and Astronomy, Swinburne University of Technology, Hawthorn, VIC 3122 Australia (e-mail: jmbrown@swin.edu.au)}
% \address[Second]{Colorado State University, 
%    Fort Collins, CO 80523 USA (e-mail: author@lamar. colostate.edu)}
% \address[Third]{Electrical Engineering Department, 
%    Seoul National University, Seoul, Korea, (e-mail: author@snu.ac.kr)}

\begin{abstract}    % 50 to 100 words
The use of mobile robotics in radioactive source seeking has become an important part of modern radiation-safety practices, supporting timely mitigation of contamination risks and helping protect public health. However, measuring radiation is often time-consuming, rendering traditional gradient-based source-seeking methods less effective due to lower sample efficiency. This paper proposes a sample-efficient Bayesian-Optimisation source-seeking strategy that utilises a heteroscedastic Gaussian process surrogate to balance exploration and exploitation. Excessive inter-sample travel is discouraged through a movement switching cost. The strategy is shown to generate sublinear regret in the source-seeking task, while simulations demonstrate its effectiveness in localising radioactive sources.
\end{abstract}


\end{frontmatter}
%===============================================================================

\section{Introduction}
Uncontrolled radioactive sources pose serious threats to the health and safety of the public, emergency responders, and defence personnel. While handheld radiation detectors are commonly used to locate these sources, their operation often exposes personnel to radiation, which can lead to serious health risks such as radiation sickness. This exposure risk can be substantially mitigated by integrating radiation detectors onto mobile robotic platforms. In particular, this work focuses on Unmanned Aerial Vehicle (UAV)-mounted detectors for efficiently locating known gamma-ray-emitting sources across large and obstructed areas. 

UAV-mounted radiation detection systems are typically used to map radiation over large areas by following preplanned flight paths. Notable examples include unmanned helicopters surveying radiation levels around the Fukushima Daiichi nuclear power plant \citep{Sanada2015, Vetter2018} and fixed-wing aircraft mapping the Chernobyl exclusion zone \citep{Connor2020}. Although preplanned trajectories are effective for producing systematic coverage maps, their inflexibility makes them poorly suited to time-critical source-seeking scenarios in which an unknown radioactive source must be quickly located within a large region. In the absence of prior environmental knowledge, these trajectories cannot adaptively concentrate measurements in areas with elevated radiation levels. Actively prioritising radiation measurements in high-activity areas is crucial in time-sensitive scenarios where rapid localisation of radioactive sources is essential to minimise radiation exposure. Consequently, source-seeking tasks require guidance strategies that use radiation measurements directly to maximise environmental information and efficiently localise the source. 

\looseness=-1
Several algorithms have been proposed for source-seeking applications. Stochastic gradient-descent methods directly use local gradient estimates of the unknown function to drive the agent toward increasing signal values \citep{stochastic_source_seeking}. In contrast, extremum seeking control \citep{KRSTIC2000313, manzie2009extremum} infers the gradient indirectly by injecting a sinusoidal perturbation signal, typically providing improved robustness and performance guarantees when compared to direct gradient descent. However, both approaches rely on repeated local measurements collected along trajectories, typically requiring extensive averaging to obtain reliable gradient estimates. This increases the number of samples required and makes these methods prone to converging to local, rather than global, extrema in the presence of multiple peaks. Therefore, they may be less suitable for radioactive source seeking where measurements may be costly and the underlying radiation intensity may be non-convex due to complex Compton scattering effects \citep{Knoll2010}.

Bayesian Optimisation (BO) provides an alternative framework for source seeking when measurements are expensive. BO constructs a surrogate model of the unknown field from noisy observations and uses a sequential acquisition strategy to select informative measurement locations, balancing exploration and exploitation, to efficiently localise the field maximum. Furthermore, BO methods typically provide asymptotic global maximisation at the expense of additional computational complexity.

A common surrogate model in BO is the Gaussian Process (GP), an infinite-dimensional generalisation of multivariate normal distributions \citep{Rasmussen2006}. GPs offer a flexible framework for function approximation and provide an estimate of approximation error through the standard deviation measure, providing insight into areas of the function that are poorly approximated. In particular, GPs have been used for radiation field estimation applications. In \cite{Silveira2018} and \cite{Khuwaileh2020}, a GP is used to reconstruct 2D dose rate profiles in radioactive environments, while in \cite{Jung2024}, a GP is used to construct a radiation intensity map from collected radiation measurements. Furthermore, GP regression can compensate for noisy measurements, providing a mechanism to compensate for the heteroscedastic noise profile found in radiation measurements \citep{Knoll2010}. 

A popular example of BO that employs a GP surrogate is given by \cite{Srinivas2012} through the Gaussian Process Upper Confidence Bound (GP-UCB) strategy. This strategy sequentially selects measurements positions that balance both exploration and exploitation of the function, as evaluated via the standard deviation and mean of the GP, respectively. Critically, \cite{Srinivas2012} establishes a sublinear regret bound when using the GP-UCB strategy, demonstrating that the average measurement of the function asymptotically approaches the true function maximum. Hence, the GP-UCB algorithm is suitable for source seeking, with the added benefit of enabling area-wide mapping through the collected data. 

\looseness=-1
However, the GP-UCB strategy proposed in \cite{Srinivas2012} does not consider movement costs incurred between consecutive measurements. These costs may be significant for systems with limited mobility or those operating within a finite time horizon. This is particularly important for UAV-mounted radiation detection, where operations are constrained by battery life. Moreover, when high spatial resolution is required, excessive manoeuvring can reduce both the number and quality of measurements, potentially exhausting the battery before the desired resolution is achieved. To address this limitation, this paper proposes and incorporates a movement switching cost into the GP-UCB algorithm for radioactive source-seeking applications. The proposed strategy generates a sublinear regret, ensuring that the maximum radiation intensity is identified. 

The paper is structured as follows. Section~\ref{sec:background} first provides a general background on radiation detection physics. A general problem formulation, including the radiation sensor model, is presented in Section~\ref{sec:problem_statement}. Bayesian optimisiation and Gaussian processes with heteroscedastic noise are then introduced in Section~\ref{sec:gp_with_heteroscedastic_noise}. The proposed BO algorithm with movement costs and corresponding theoretical guarantees are presented in Section~\ref{sec:gp_ucb_with_switching}, with Section~\ref{sec:simulation_results} then providing empirical demonstration of the algorithm in simulated radioactive source-seeking scenarios. The paper is then concluded in Section~\ref{sec:conclusion}.

\section{Background}\label{sec:background}
\subsection{Radiation Detection Physics}
Gamma rays emitted by decaying radionuclides have quantised energies that are characteristic of the specific radionuclide. The precise measurement of these energies with a radiation detector, such as a scintillation-based detector, allows for isotope-specific radioactive source identification. Scintillation-based detectors typically consists of a scintillation crystal optically bonded to a Silicon PhotoMultiplier (SiPM) array. As shown in Figure~\ref{fig:gamma_ray_interactions}a, gamma rays interact inelastically with the scintillation crystal primarily through the three mechanisms:
\begin{description}
    \item[Photoelectric absorption:] A bound atomic electron absorbs the full energy of an incident gamma ray and is ejected from the atom.
    \item[Compton scattering:] An incident gamma ray scatters with an electron such that momentum is conserved.
    \item[Pair production:] For gamma-ray energies above 1.022 MeV, an incident gamma ray in the electric field of an atom can produce an electron-positron pair. The positron can annihilate with an electron in the surrounding material and produce two characteristic 511 keV gamma rays, known as annihilation gamma rays.
\end{description}
The dominant interaction depends on the energy of the incident gamma ray and the effective atomic number of the material in which the interaction occurs.

\begin{figure}[h!]
    \centering
    \includegraphics[width=0.59\columnwidth]{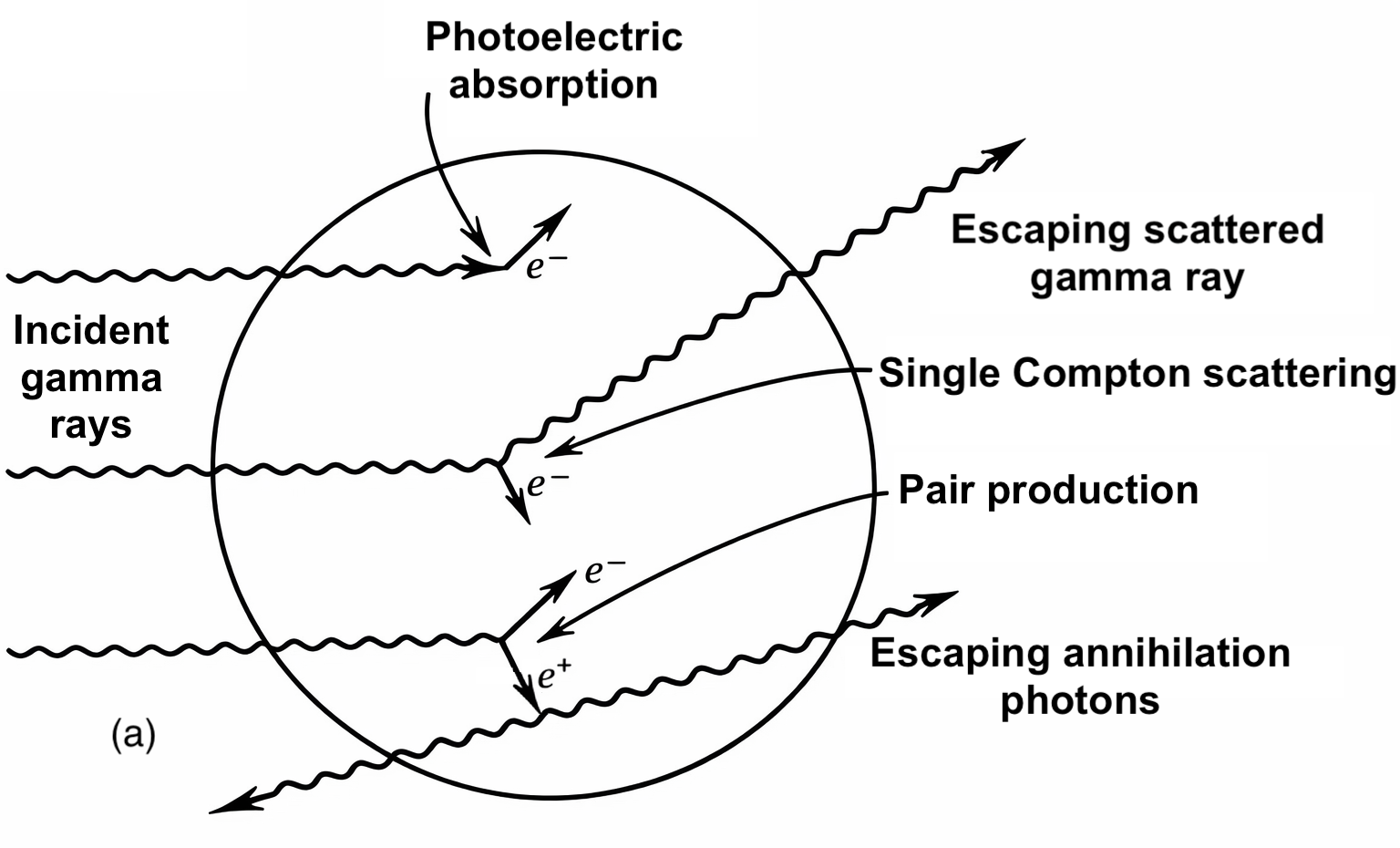}
    \includegraphics[width=0.39\columnwidth]{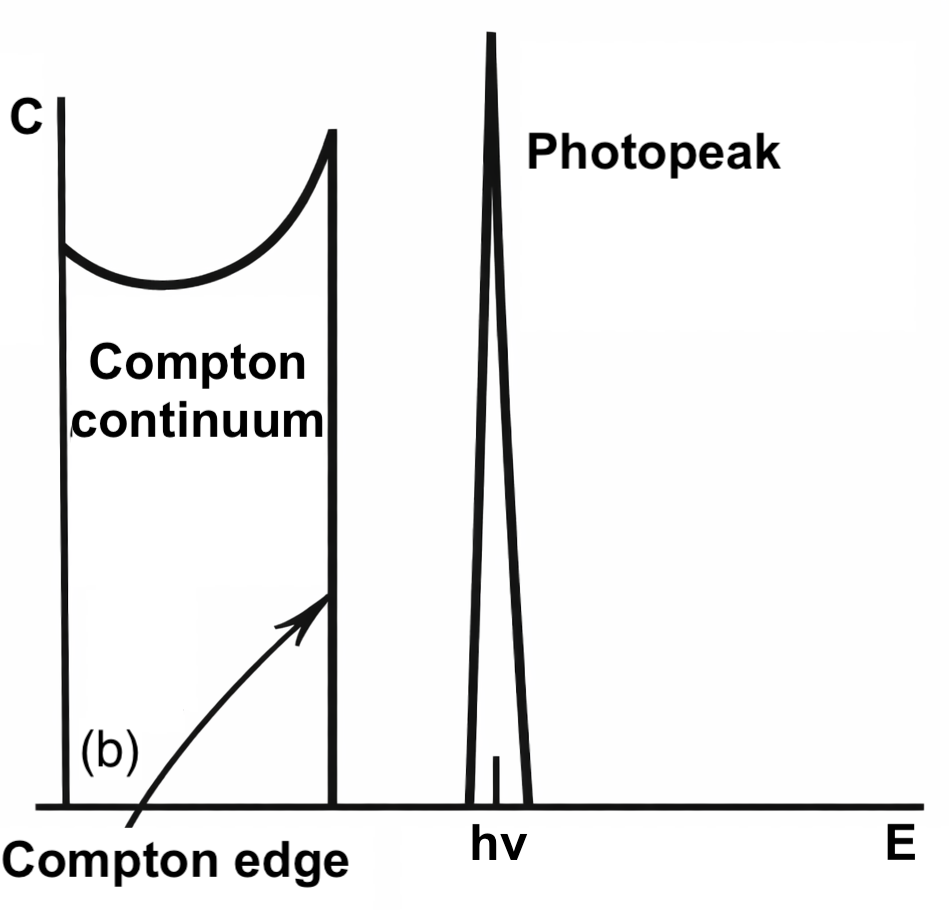}
    \caption{Gamma-ray interactions with a scintillation crystal (a). Energy spectrum resulting from the gamma-ray interactions within the scintillation crystal (b). Images adapted from \cite{Knoll2010}.}\label{fig:gamma_ray_interactions}
\end{figure}
\unskip

The interaction between the gamma ray and scintillation crystal results in the partial or complete transfer of the energy of the gamma ray into the electrons in the crystal lattice. These excited electrons de-excite to the valence band via the emission of optical photons from the scintillation material. The optical photons reach the SiPM array in the detector, producing an analog pulse with amplitude proportional to the number of optical photons detected. By calibrating the scintillation-based detector with known radioactive sources, an energy spectrum can be produced from the electrical signal over the detector measurement time. As shown in Figure~\ref{fig:gamma_ray_interactions}b, the energy spectrum is a plot of the number of gamma rays detected within discrete energy windows. 

A feature of the energy spectrum is the photopeak, which corresponds to the gamma rays that are photoelectrically absorbed by the scintillation crystal. The energy of the incident gamma ray can be inferred from the photopeak, while the number of gamma rays at that energy can be determined from the counts under the photopeak. The width of the photopeak is characteristic of the scintillation material. A higher energy resolution corresponds to a narrower photopeak width and indicates that the detector can more precisely measure the energy of the incident gamma ray. By summing the photopeak count for all characteristic gamma rays, the \textit{radionuclide count} can be calculated, which is the accumulation of all photoelectrically absorbed gamma rays emitted by the radionuclide.

\section{Problem Formulation}\label{sec:problem_statement}
This work seeks to localise an unknown radioactive point source by leveraging radionuclide count measurements collected with a UAV-mounted scintillation-based detector. This problem is motivated by previous search operations, such as the localisation of an uncontained $^{137}$Cs source in Western Australia \citep{ABCNews2023}. Let $\mathbb{N}$ and $\mathbb{R}$ denote the set of natural and real numbers, define $\mathbb{R}^d$ as the $d$-dimensional Euclidean space, and let $||x||_p$ denote the $p$-norm of $x \in \R^d$. Let $\mathcal{R} \subset \R^3$ denote the total space of interest within the source-seeking task and let $Q \doteq \{q \in \mathcal{R}  \ | \ [0,0,1]^Tq = 0\}$ denote the ground plane of $\mathcal{R}$ in which the radioactive point source is located. The point source, with true location at $q \in Q$, emits $N \in \mathbb{N}$ detectable characteristic gamma rays with an activity $\lambda \in \Rpos$ in becquerels (Bq), defined as one expected nuclear decay per second. See \cite{Miller2025} for further details. 

Let $\Omega \doteq \{q \in \mathcal{R} \ | \ [0,0,1]^Tq = h\}$ denote the UAV operating plane located in $\mathcal{R}$ at height $h\in \Rpos$. At iteration $t \in \mathbb{N}$ denote the set of all discrete measurement locations and corresponding measurement values as $\mathbf{x}_t = (x_1, ... , x_t) \in \Omega^t$ and $\mathbf{y}_t = (y(x_1), ... , y(x_t)) \in \mathbb{N}^t$ such that $y(x) \in \mathbb{N}$ denotes the measurement at position $x \in \Omega$. This work aims to drive the UAV towards the position of the unknown source $q \in Q$ by leveraging collected information $\mathbf{x}_t, \mathbf{y}_t$. Hence, we consider the sequential discrete path-planning task of the UAV such that at iteration $t \in \mathbb{N}$ measurement points are sampled according to the strategy $f: \Omega^t \times \mathbb{N}^t \to \Omega$ such that
\begin{align*}
    x_{t+1} \doteq f(\mathbf{x}_t, \mathbf{y}_t).
\end{align*}
We next formulate a radiation sensing model to assist in the design of a suitable sampling strategy $f(\cdot,\cdot)$. Measurements collected by the detector at $x \in \Omega$ over $\tau \in \Rpos$ seconds are represented by the sensor model $y: \Omega \to \mathbb{N}$. Here, $y(x)$ is the sum of the detected counts originating from the source $n(x) \in \mathbb{N}$ and background radiation, represented by the background Counts Per Second (CPS) function $B:\Omega\to\mathbb{N}$. The sensor model is given by
\begin{align}
    \label{eq:sensor_model}
    y(x) \doteq n(x) + B(x)\tau.
\end{align}

The number of source counts detected during a measurement period is influenced by both the quantity of gamma rays emitted from the source and the detector efficiency. An experimentally validated radionuclide detection efficiency model is provided by \cite{Miller2025}. For a point source located at $q \in Q$ with $N$ detectable characteristic gamma rays indexed by $k\in\{1, ..., N\}$, the radionuclide detection efficiency model $\hat{\phi}:\Omega \to \mathbb{R}$ is given by
\begin{equation}
   \hat{\phi}(x) \doteq \frac{r\eta}{2\pi||x-q||_2^2}\sum_{k=1}^{N}\alpha_k\epsilon_k\exp(-\mu_{\text{air}, k}||x-q||_2),\label{eq:radiation_field_model}
\end{equation}
where $r, \eta \in \Rpos$ are the radius and height of a cylindrical scintillation crystal in the radiation detector. For the $k$th gamma ray, $\alpha_k$ is the branching ratio, $\epsilon_k$ is the intrinsic peak efficiency, and $\mu_{\text{air}, k}$ is the air absorption coefficient.\footnote{The branching ratio $\alpha_k$ represents the probability that a radionuclide decays and emits a gamma ray with energy $E_k$. The likelihood that the gamma ray is photoelectrically absorbed within the scintillation crystal is captured by the intrinsic peak efficiency $\epsilon_k$. Finally, the air absorption coefficient $\mu_{\text{air}, k}$ quantifies the proportion of gamma rays absorbed or scattered by air molecules before reaching the detector.}

Gamma-ray emission induced by radioactive decay follows a Poisson distribution \citep{Knoll2010}. At $x \in \Omega$, the expected number of detected counts from a radionuclide point source with activity $\lambda$ over a period of $\tau$ seconds is $\mathbb{E}(n(x)) = \tau \lambda \hat{\phi}(x)$. Therefore, assuming isotropic particle emission, the count distribution follows
\begin{equation}
    n(x) \sim P[x](n) \doteq \frac{\left(\tau \lambda \hat{\phi}(x)\right)^n \ e^{-\tau \lambda \hat{\phi}(x)}}{n!}. \label{eq:poisson_distribution}
\end{equation}
For sufficiently high expected counts, $\hat{\phi}(x)\tau \lambda > 25$ according to \cite{Knoll2010}, $P(n)$ is well approximated by a normal distribution. Therefore, assuming sufficiently high activity, the source count distribution can be approximated as
\begin{align}
    \label{eq:gaussian_sensor_model}
    n(x) \sim \mathcal{N}(\tau \lambda \hat{\phi}(x),  \tau \lambda \hat{\phi}(x)).
\end{align}
Similarly, the background count $B(x)\tau$ in \eqref{eq:sensor_model} can also be approximated by a normal distribution assuming $B(x)\tau>25$ for all $x \in \Omega$ \citep{Knoll2010}. As demonstrated by \cite{Miller2025}, the background for a fixed altitude can be approximated as constant in the x- and y-direction over the operating regime of interest. The background CPS can therefore be represented by $B(x)\sim \mathcal{N}(b,b)$ for all $x \in \Omega$. Here, $b\in\mathbb{R}$ is the background CPS at height $h$, accumulated at each characteristic gamma-ray energy. Separating \eqref{eq:sensor_model} into its mean and noise term using \eqref{eq:gaussian_sensor_model} yields
\begin{equation}
    y(x) = \phi(x) + v(x), \qquad v(x) \sim \mathcal{N}(0, \sigma^2(x)),\label{eq:detector_measurement}
\end{equation}
where
\begin{align}
    \phi(x) & \doteq  \tau \lambda \hat{\phi}(x)+ \tau b,\label{eq:measurement_model} \\ \sigma^2(x) & \doteq \tau\lambda \hat{\phi}(x)+\tau b.\label{eq:noise_variance}
\end{align}
Here, $\phi:\Omega\to\mathbb{R}$ represents the scalar field of the mean radionuclide count, hereafter referred to as the radiation field, while $\sigma^2:\Omega\to\mathbb{R}$ denotes the corresponding measurement noise variance. The measurement noise is heteroscedastic due to its position dependence. Given \eqref{eq:radiation_field_model} and \eqref{eq:measurement_model}, it is then routine to show that for a single point source, the point $x^* \in \Omega$ closest to $q \in Q$ satisfies
\begin{align}
    \label{min_equiv}
    x^* = \argmin{x \in \Omega} \|x-q\|_2 = \argmax{x \in \Omega} \ \phi(x),
\end{align}
i.e., the expected radionuclide count measurement is maximised when the detector is nearest to the source. Note the equivalence in \eqref{min_equiv} is not guaranteed in general for multiple point sources.

\section{Bayesian Optimisation with Heteroscedastic Gaussian Processes}\label{sec:gp_with_heteroscedastic_noise}
Equation \eqref{min_equiv} demonstrates that source seeking may be recast as an optimisation problem over $\phi(\cdot)$. However, measurement collection times may be expensive (on the order of minutes per sample for low-activity sources) to ensure statistical significance within the measured energy spectrum. Compton scattering with the surrounding environment may also introduce secondary local maxima and added noise into $\phi(\cdot)$ within real-world activities. Hence, methods such as stochastic gradient descent \citep{stochastic_source_seeking} and extremum seeking control \citep{manzie2009extremum} may be unsuitable in efficiently locating $q$ within the radioactive source-seeking application.

\sloppy
An alternative approach that better leverages the heteroscedastic Gaussian structure of \eqref{eq:detector_measurement} is that of Bayesian Optimisation (BO), a sequential global optimisation method that aims to efficiently maximise some unknown function when measurements are expensive. This is achieved by constructing and analysing a probabilistic surrogate of the function, typically modelled via a Gaussian Process (GP). A GP, denoted $\mathcal{GP} (\mu(\cdot), k(\cdot, \cdot))$, can be considered a distribution of functions and is an infinite dimensional analogue of a multivariate Gaussian distribution that is fully defined by a mean function $\mu:\R^d \to \R$ and covariance function $k:\R^d \times \R^d \to \Rpos$. The selection of covariance function $k(\cdot,\cdot)$ is significant as it sets the class of functions that can be represented by the GP, see \cite{Rasmussen2006} for further details. Several covariance kernels have been used in GPs for radiation detection applications. These include the radial basis function in \cite{Khuwaileh2020}, Mat\'ern 3/2 in \cite{Silveira2018}, and inverse square in \cite{Jung2024}. 

We now consider radioactive source seeking as a BO problem. Assume the selection of $k(\cdot,\cdot)$ such that $\phi(\cdot) \sim \mathcal{GP} (0, k(\cdot, \cdot))$, where a prior mean of zero is chosen without loss of generality. At iteration $T$, the set of previous measurements $\mathbf{x}_T, \mathbf{y}_T$ can be used to condition the  GP and provide a lower variance posterior estimate in a process known as GP regression. First, the joint distribution of the training set $\mathbf{x}_T, \mathbf{y}_T$, and the unknown function $\phi(\cdot)$ evaluated at some test point $\tilde{x} \in \Omega$ is 
\begin{equation}
    \begin{pmatrix}
        \textbf{y}_T \\
        \phi(\tilde{x})
    \end{pmatrix} \sim \mathcal{N}\Bigg(
    \begin{pmatrix}
        \boldsymbol{0}_T \\
        0
    \end{pmatrix}, \begin{pmatrix}
        K_T+Q_T & \boldsymbol{k}_T(\tilde{x}) \\
        \boldsymbol{k}_T^\top(\tilde{x}) & k(\tilde{x}, \tilde{x})
    \end{pmatrix} \Bigg),\label{eq:GP_prior}
\end{equation}
where $\boldsymbol{0}_T$ is a $T\times 1$ zero matrix, $K_T=[k(x, x^\prime)]_{x, x^\prime \in \textbf{x}_T}$, $Q_T=\diag(\sigma^2(x_1),...,\sigma^2(x_T))$ with $\sigma^2(\cdot)$ as per \eqref{eq:noise_variance}, and $\boldsymbol{k}_T(x)=(k(x_1, x), ... , k(x_T, x))^\top$. 
The posterior over the unknown function $\phi(\cdot)\sim\mathcal{GP}(\mu_T(\cdot), \sigma_T^2(\cdot))$, conditioned on the $T$ measurements, has mean and covariance functions
\begin{align}
    \mu_T(x) &= \boldsymbol{k}_T^\top(x)(K_T+Q_T)^{-1}\textbf{y}_T,\label{eq:posterior_mean}\\
    \sigma^2_T(x) & = k(x, x)- \boldsymbol{k}_T^\top(x)(K_T+Q_T)^{-1}\boldsymbol{k}_T(x), \label{eq:posterior_std}
\end{align}
which has been adapted from \cite{Rasmussen2006} for heteroscedastic noise. 

At iteration $T$, the conditioned mean $\mu_T(\cdot)$ represents the current best approximation of $\phi(\cdot)$ given $\mathbf{x}_T, \mathbf{y}_T$. In parallel, $\sigma_T(\cdot)$ quantifies the relative uncertainty of $\phi(\cdot)$ over $\Omega$. This uncertainty quantification is a key strength of GP-based BO techniques, as it enables efficient sampling of regions with high expected information gain.

A popular algorithm that leverages the posterior mean and covariance for BO is the Gaussian Process Upper Confidence Bound (GP-UCB) algorithm from \cite{Srinivas2012}. The GP-UCB strategy is, for $0\leq t \leq T$,
\begin{equation}
    x_{t} \doteq \argmax{x\in\Omega}\big[\mu_{t-1}(x)+\beta_{t}^{1/2}\sigma_{t-1}(x)\big],\label{eq:gp_ucb_algorithm}
\end{equation}
where $\beta_{t}\in\mathbb{R}_{>0}$ is a time-varying parameter and $\mu_{t-1}(x)$ and $\sigma_{t-1}(x)$ denote the posterior GP mean and standard deviation, given by \eqref{eq:posterior_mean} and \eqref{eq:posterior_std} with homoscedastic noise, i.e., $\sigma^2(x)$ is constant for all $x\in\Omega$. One interpretation of \eqref{eq:gp_ucb_algorithm} is that it greedily selects a position $x_{t}$ that maximises the upper confidence bound on the GP posterior. Exploration and exploitation are balanced via the $\beta_{t}^{1/2}\sigma_{t-1}(x)$ and $\mu_{t-1}(x)$ terms, respectively. The growth of $\beta_{t}^{1/2}$ with $t$ promotes exploration by encouraging sampling in high-uncertainty regions where the true maximum may remain undetected.

\section{Bayesian Optimisation with \\ Movement Costs}\label{sec:gp_ucb_with_switching}
Notably, \eqref{eq:gp_ucb_algorithm} does not prioritise nearby sample points and can result in large displacements between consecutive measurements. In many applications, it is preferable to minimise unnecessary movement to conserve resources, such as UAV battery life.

To avoid large movements between consecutive measurement positions, \cite{Marchant2012} introduced a distance-dependent switching cost into a variant of GP-UCB known as the Gaussian Process Distance Upper Confidence Bound (GP-DUCB) algorithm. This switching cost is constant and proportional to Euclidean distance between points. However, theoretical guarantees for GP-DUCB are absent. To address these limitations we propose an alternate GP-DUCB strategy with time-dependent switching costs, and demonstrate corresponding sublinear regret guarantees. The strategy is, for $0 \leq t \leq T$, 
\begin{align}
    x_{t} \doteq \argmax{x\in\Omega}\left [\mu_{t\!-\!1}(x) + \underbrace{\beta_t^{1/2} \sigma_{t\!-\!1}(x)\left(1 \! - \! 2\rho_t ||x\! -\!x_{t\!-\!1}||^2_2\right)}_{\text{exploration term}}\right] 
    ,\label{eq:gp_ducb_algorithm}
\end{align}
where for  $0<C < \infty$, $t \mapsto \rho_t \in \mathcal{L}^\infty(\mathbb{N};(0, C])$ i.e $\rho_t \leq C$ for all $t$. Note, for $\|x-x_{t-1}\|_2 \geq \bar{r}(\rho_t) \doteq \left(2\rho_t\right)^{-1/2}$, the exploration term in \eqref{eq:gp_ducb_algorithm} becomes non-positive. Therefore, the quantity $\bar{r}(\rho_t)$ is dubbed the soft exploration radius since exploration-focused steps (i.e. steps dominated by the exploration term) are disincentivized outside of $\bar{r}(\rho_t)$. This aims to reduce sudden agent ``teleportation" when exploring the space. However, this penalisation only applies to the exploration term and large movements may still occur when exploiting locations with high posterior mean (i.e. $\mu_{t\!-\!1}(x)$ is large and $\sigma_{t-1}(x)$ is small). In addition, the soft exploration radius $\bar{r}(\rho_t)$ is directly influenced by selection of the sequence $\rho_t$. A possibility is to use a decreasing $\rho_t$  such that $\bar{r}(\rho_t)$ increases over time to encourage exploration of distant uncertain regions. The GP-DUCB algorithm is shown in pseudocode in Algorithm~\ref{alg:gp_ucb}. 

\begin{algorithm}
\caption{GP-DUCB}\label{alg:gp_ucb}
\quad\textbf{Input:} Region $\Omega$, $x_0\in\Omega$, $\mu_0$, $\sigma_0$, and kernel $k(\cdot,\cdot)$
\begin{algorithmic}
\For{$t = 1$ \textbf{to} $T$}
    \State Choose $x_t$ via \eqref{eq:gp_ducb_algorithm}.
    \State Obtain $y_t$ via \eqref{eq:detector_measurement}.
    \State Compute $\mu_{t}(\cdot)$ and $\sigma_{t}(\cdot)$ via \eqref{eq:posterior_mean} and \eqref{eq:posterior_std}.\label{line:update_gp}
\EndFor
\end{algorithmic}
\end{algorithm}

\subsection{Regret Analysis}\label{sec:sublinear_regret}
We now demonstrate that the GP-DUCB will globally optimise $\phi(\cdot)$. To support this analysis, it is useful to restate several useful bounds from \cite{Srinivas2012}. First, for $x_t \in D \subset \R^d$ where $D$ is compact for all $t \in \mathbb{N}$, Lemma 1 establishes a high-probability upper confidence bound on the conditioned GP estimate at sample locations.
\vspace{2mm}

\begin{lemma}[Lemma 5.5 in \cite{Srinivas2012}]
Pick $\delta\in(0, 1)$ and set $\beta_t=2\log(\pi_t/\delta)$, where $\sum_{t\geq 1}\pi_t^{-1}=1$, $\pi_t>0$. Then
\begin{equation*}
    |\phi(x_t)-\mu_{t-1}(x_t)|\leq \beta_t^{1/2}\sigma_{t-1}(x_t), \ \forall t \geq 1,
\end{equation*}
holds with probability $\geq 1 - \delta$. 
\end{lemma}
Lemma 2 then establishes a similar high-probability upper confidence bound for some arbitrary $x^* \in D$. This is complicated by the infinite cardinality of $D$ and relies on considering a finite-cardinality subset $D_t \subset D$, with size ($\tau_t)^d$ for $\tau_t\in\mathbb{R}$, which becomes dense in $D$ as $t \to \infty$.

\vspace{2mm}
\begin{lemma}[Lemma 5.7 in \cite{Srinivas2012}]
    Consider $x_t \in D \subset \R^d$. Pick $\delta\in(0, 1)$ and constants $a,b > 0$. Set $\beta_t=2\log(2\pi_t/\delta)+4d\log(dtbr\sqrt{\log(2da/\delta)})$, where $\sum_{t\geq 1}\pi_t^{-1}=1$, $\pi_t>0$. Let $\tau_t=dt^2br\sqrt{\log(2da/\delta)}$. Let $[x^*]_t$ denote the closest point in $D_t$ to $x^*$. Then
\begin{equation*}
    |\phi(x^*)-\mu_{t-1}([x^*]_t)|\leq\beta_t^{1/2}\sigma_{t-1}([x^*_t])+\frac{1}{t^2}, \ \forall t \geq 1,
\end{equation*}
holds with probability $\geq 1-\delta$. 
\end{lemma}
Finally, Lemma 3 extends a key high-probability bound in \cite{Srinivas2012} to the heteroscedastic setting.
\vspace{2mm}

\begin{lemma}[Adapted from Lemma 5.4 in \cite{Srinivas2012} for the setting]
    Pick $\delta\in(0, 1)$ and let $\beta_t=2\log(|D|\pi_t/\delta)$, where $\sum_{t\geq 1}\pi_t^{-1}=1$, $\pi_t>0$. Then, the following holds with probability $\geq 1-\delta$
\begin{equation*}
    \sum^T_{t=1}4\beta_t\sigma^2_{t-1}(x_t)\leq C_1\beta_T\gamma_T, \ \forall T \geq 1,
\end{equation*}
where $C_1 = \max_{1< t \leq T}\big(8/\log(1+\sigma^{-2}(x_t))\big)$. 
\end{lemma}

We now consider a regret formulation of the radioactive source-seeking problem. Consider the general static cumulative regret with switching cost, as adapted from \cite{Goel2019}, given a cost maximisation objective
\begin{equation*}
    R_T = \sum^T_{t=1}g(x^*_t) - c(x^*_t, x^*_{t-1}) - \Big[\sum_{t=1}^Tg(x_t) - c(x_t, x_{t-1})\Big],
\end{equation*}
where $g:\Omega \rightarrow \mathbb{R}$ is the hitting cost, $c: \Omega \times \Omega \rightarrow \mathbb{R}$ is the switching cost, and $x^*_t\in \Omega$ is the optimal position for $t \in \{1, ..., T\}$. In the context of radioactive source seeking, the hitting cost represents the loss incurred by not evaluating the function at its maximiser $x^*_t = x^*$ as per (\ref{min_equiv}) for $t\in\{1, ..., T\}$, implying $g(x) = \phi(x)$ for all $x \in \Omega$. The instantaneous regret formulation for a spatially static unknown function is
\begin{equation}
    r_t = \phi(x^*) - \phi(x_t) + c(x_t, x_{t-1}),\label{eq:instant_regret}
\end{equation}
where the movement switching cost is
\begin{equation}
c(x_t, x_{t-1})=2\beta_t^{1/2}\sigma_{t-1}(x_t)\rho_t||x_t-x_{t-1}||^2_2,
\end{equation}
as per (\ref{eq:gp_ducb_algorithm}), and the cumulative regret is
\begin{align}
    \label{cummul_reg}
    R_T = \sum^T_{t=1} r_t.
\end{align}
It is now possible to state the main result of the paper.
\vspace{2mm}

\begin{theorem}[Main Result]
\label{Theorem1}
Let $\Omega \subset [0, r]^d$ be compact and convex, $d\in \mathbb{N}$, $r>0$. Suppose that the kernel $k(x, x^\prime)$ satisfies the following high probability bound on the derivatives of GP sample paths $\phi(\cdot)$ for all $j = 1, ..., d$
\begin{equation}
    \text{Pr}\Big\{\sup_{x\in \Omega}\Big|\frac{\partial \phi}{\partial x_j}\Big| > L \Big\} \leq a \exp(-(L/b)^2), \ \label{eq:lipschitz_gp_condition}
\end{equation}
for some constants $a, b, L > 0$. Pick $\delta \in (0, 1)$, and define
\begin{align*}
    \beta_t &= 2\log(2t^2\pi^2/(3\delta)) + 2d \log\Big(t^2dbr\sqrt{\log(4da/\delta)}\Big),\\
    \tau_t &= dt^2br\sqrt{\log(2da/\delta)},
\end{align*}
and $0\leq\rho_t\leq C$ is bounded above for all $t\in\{1, ..., T\}$. Then, under the dynamics in \eqref{eq:gp_ducb_algorithm}, the cumulative regret $R_T$, as per \eqref{cummul_reg}, is bounded by a sublinear function of $T$ such that, for $0 < \delta < 1$,
\begin{equation}
    \label{eq:sublinear_regret}
    \text{Pr}\Big\{R_T \leq (1+Cr^2d)\sqrt{C_1T\beta_T\gamma_T} + \frac{\pi^2}{6} \ \ \forall T \geq 1\Big\} \geq 1 - \delta,
\end{equation}
where $C_1 = \max_{1< t \leq T}\big(8/\log(1+\sigma^{-2}(x_t))\big)$, and $\gamma_T$ is the maximum information gain after $T$ rounds.
\end{theorem}

\begin{proof}
The proof follows a similar structure to \cite{Srinivas2012}. Consider the sequence of discretisations $\Omega_t\subset \Omega$ of size $(\tau_t)^d$, such that for every $x\in \Omega$, $||x-[x]_t||_1 \leq rd/\tau_t$. Using Lemmas 1 and 2, and the GP-DUCB strategy in (\ref{eq:gp_ducb_algorithm}), the instantaneous regret as per (\ref{eq:instant_regret}) can be bounded as
\begin{equation*}
    r_t \leq 2\beta_t^{1/2}\sigma_{t-1}(x_t)\big(1+\rho_t||[x^*]_t-x_{t-1}||^2_2\big)+ \frac{1}{t^2}.
\end{equation*}
By taking the worst-case bound $||[x^*]_t-x_{t-1}||_2 \leq r\sqrt{d}$ and using the Cauchy-Schwarz inequality, the cumulative regret can be upper bounded by 
\begin{equation*}
    R_T \leq \Big[\sum_{t=1}^T (1+ \rho_tr^2d)^2\Big]^{1/2}\Big[\sum_{t=1}^T (2\beta_t^{1/2}\sigma_{t-1}(x_t))^2\Big]^{1/2}+\frac{\pi^2}{6},
\end{equation*}
where the $\frac{\pi^2}{6}$ term is due to Euler's evaluation of the Basel sum. By recalling that $\rho_t\leq C$ for all $1\leq t \leq T$, and applying Lemma 3, the upper confidence bound on $R_T$ can be simplified as per the theorem statement, i.e. \eqref{eq:sublinear_regret}. 
\end{proof}

Theorem~\ref{Theorem1} shows that the cumulative regret of the proposed GP-DUCB strategy grows sublinearly with the number of measurements.
\vspace{2mm}

\begin{corollary}
\label{corollary:convergence_to_optimal}
Under \eqref{eq:gp_ducb_algorithm}, for $0<\delta<1$, with probability at least $1-\delta$,
    \begin{align*}
         \lim_{T\to \infty} \frac{1}{T}\sum^T_{t=1} \left[c(x_t,x_{t-1})\right] =  \lim_{T\to \infty} \frac{1}{T}\sum^T_{t=1} \left[\phi(x_t)-\phi(x^*) \right] = 0.
    \end{align*}
\end{corollary}
\begin{proof}
    It follows directly from \eqref{eq:sublinear_regret} in Theorem~\ref{Theorem1} that
    \begin{align*}
        \lim_{T\to \infty} \frac{1}{T}\sum^T_{t=1} \left[ \phi(x_t)-\phi(x^*) -c(x_t,x_{t-1})\right] = 0.
    \end{align*}
    Noting that $\phi(x) - \phi(x^*) \leq 0$ for all $x \neq x^*$
    \begin{align*}
         \lim_{T\to \infty} \frac{1}{T}\sum^T_{t=1} \left[c(x_t,x_{t-1})\right] =  \lim_{T\to \infty} \frac{1}{T}\sum^T_{t=1} \left[\phi(x_t)-\phi(x^*) \right]\leq 0.
    \end{align*}
    However, this expression must hold with equality as $c(\cdot, \cdot)$ is a non-negative function. The corollary assertion then follows.
\end{proof}
Therefore, as per Corollary \ref{corollary:convergence_to_optimal}, the UAV becomes asymptotically stationary (on average) and the average measurement approaches the global maximum of $\phi(\cdot)$ under the GP-DUCB strategy in \eqref{eq:gp_ducb_algorithm}. Note, the sublinear regret bound in Theorem 1 still admits exploration away from $x^*$ under the strategy, however the frequency of such exploratory deviations decreases over time. 

\section{Simulation Results}\label{sec:simulation_results}
This section empirically demonstrates radioactive source seeking using the GP-DUCB algorithm in (\ref{eq:gp_ducb_algorithm}) through simulation. Simulations were designed to emulate the flight trials of the UAV-mounted radiation detection system in \cite{Miller2025} as the radiation field model in (\ref{eq:measurement_model}) was experimentally validated in this work. The simulation environment considers the problem of localising a radioactive source within the rectangular region defined by $\mathcal{R} = [0, 20 ] \text{ m} \times [0, 20] \text{ m}\times [0, 3] \text{ m}$. The scintillation-based detector measurement model in (\ref{eq:detector_measurement}) was used to simulate the radiation field generated by a 1000 MBq $^{137}$Cs source positioned at $q=(15 \text{ m}, 15 \text{ m})$ from a flight height of $h=3$ m. Radioactive decay of $^{137}$Cs emits a single ($N = 1$) characteristic 662 keV gamma ray with a probability of $\alpha_1=85.10 \%$ \citep{NNDC2025}. The corresponding air absorption coefficient is $\mu_{\text{air}, k}=9.95 \times 10^{-5}$ cm$^{-1}$ \citep{NISTXCOM}.

This work considers a NaIL scintillation crystal (95\% $^6$Li enriched lithium co-doped NaI:Tl) with radius $r = 2.54$~cm and height $\eta = 5.08 $ cm in the measurement model of (\ref{eq:detector_measurement}). The intrinsic peak efficiency for the 662 keV gamma ray was $\epsilon_1=24.42\%$ \citep{Miller2025}. The radiation field was simulated for a measurement time of $\tau = 10$ s and the Gaussian noise $v(\cdot)$ was generated using the noise variance $\sigma^2(\cdot)$ according to \eqref{eq:noise_variance}. A constant background radiation of $b = 3.07$~counts was used. This value was measured experimentally at 3 m flight height using a NaIL scintillation crystal \citep{Miller2025}. The background data was imported and rescaled using the measurement time of $\tau = 10$ s. Counts around the 662 keV photopeak were calculated by summing the background counts within three standard deviations corresponding to an energy resolution of 7.1\%. 

The heteroscedastic GP was implemented in Python using the GPy Python library (v1.13.2) from \cite{GPy2024}. Simulations used the Mat\'ern 5/2 covariance kernel, which satisfies the GP derivative condition in Theorem \ref{Theorem1}, i.e., \eqref{eq:lipschitz_gp_condition} \citep{Srinivas2012}. For $x, x^\prime \in \Omega$, the kernel is defined as
\begin{equation*}
\begin{aligned}
k(x, x') &\doteq \Big(1+\frac{\sqrt{5}|| x-x' ||}{l} +\frac{5||x-x'||^2}{3l^2}\Big) \\
&\quad \times \exp\Big(-\frac{\sqrt{5}|| x-x' ||}{l}\Big),
\end{aligned}
\end{equation*}
where $l$ is the length scale. The GP-DUCB algorithm was implemented following Algorithm~\ref{alg:gp_ucb} with $\delta=0.1$, and was initialised at the platform position of $x_0 = (5 \text{ m}, 2.5 \text{ m})$ with a corresponding radiation measurement. For computational efficiency, the GP posterior was computed using the GPy GPHeteroscedasticRegression class. The algorithm was executed for $t\in\{1, ..., T\}$ iterations, where $T$ was determined by the UAV battery-life constraint 
\begin{equation}
    \tau + \sum_{t=1}^T\Big[\frac{1}{u}||x_t-x_{t-1}||_2 + \tau \Big]\leq \text{12 minutes},\label{eq:stopping_condition}
\end{equation}
with a flight speed of $u=1$ m s$^{-1}$. This stopping condition represents a realistic operational timescale, enabling a comparison between flight trajectories with low movement switching costs and those generated under higher switching costs.

\begin{figure}[h!]
    \centering
    % ---------- Row 1 ----------
    \begin{subfigure}{\linewidth}
        \centering
        \begin{subfigure}{0.49\linewidth}
            \centering
            \includegraphics[width=\linewidth]{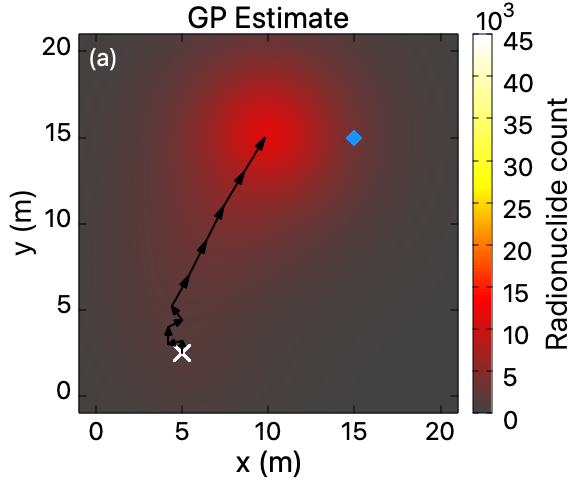}
        \end{subfigure}
        \hfill
        \begin{subfigure}{0.49\linewidth}
            \centering
            \includegraphics[width=\linewidth]{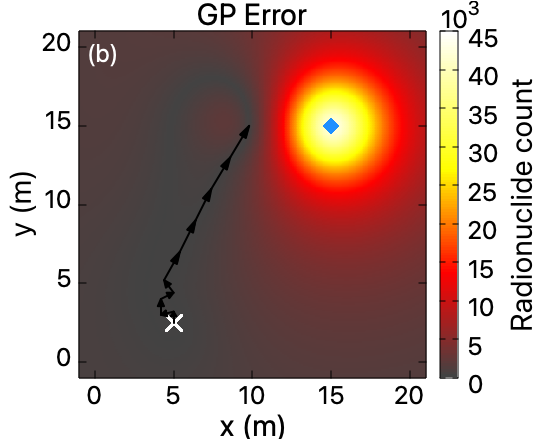}
        \end{subfigure}
        \caption*{2 minute flight time.}
    \end{subfigure}

    \vspace{0.75em}

    % ---------- Row 2 ----------
    \begin{subfigure}{\linewidth}
        \centering
        \begin{subfigure}{0.49\linewidth}
            \centering
            \includegraphics[width=\linewidth]{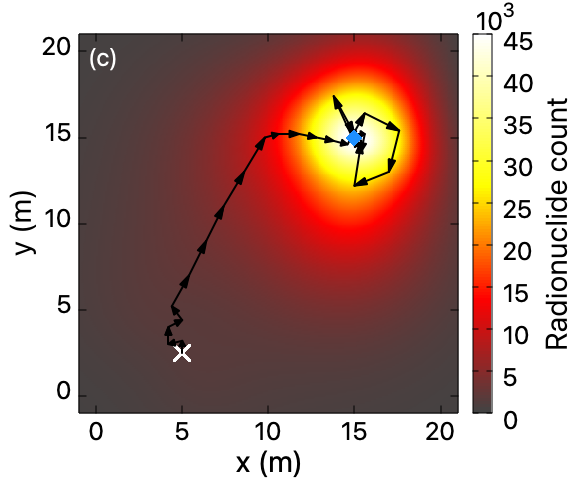}
        \end{subfigure}
        \hfill
        \begin{subfigure}{0.49\linewidth}
            \centering
            \includegraphics[width=\linewidth]{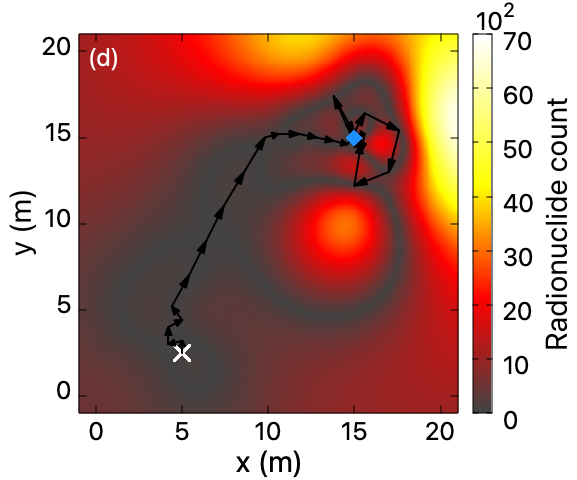}
        \end{subfigure}
        \caption*{6 minute flight time.}
    \end{subfigure}

    \vspace{0.75em}

    % ---------- Row 3 ----------
    \begin{subfigure}{\linewidth}
        \centering
        \begin{subfigure}{0.49\linewidth}
            \centering
            \includegraphics[width=\linewidth]{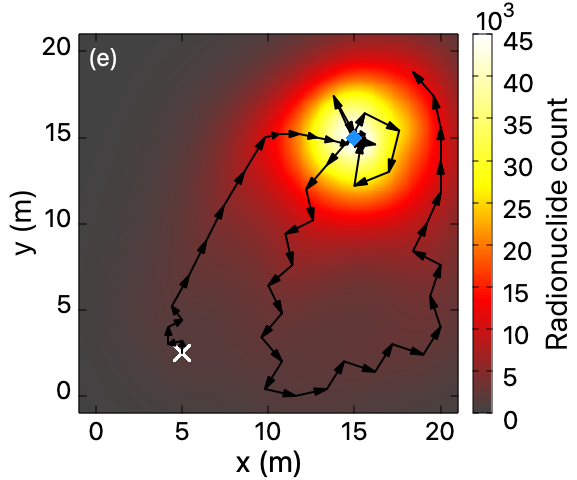}
        \end{subfigure}
        \hfill
        \begin{subfigure}{0.49\linewidth}
            \centering
            \includegraphics[width=\linewidth]{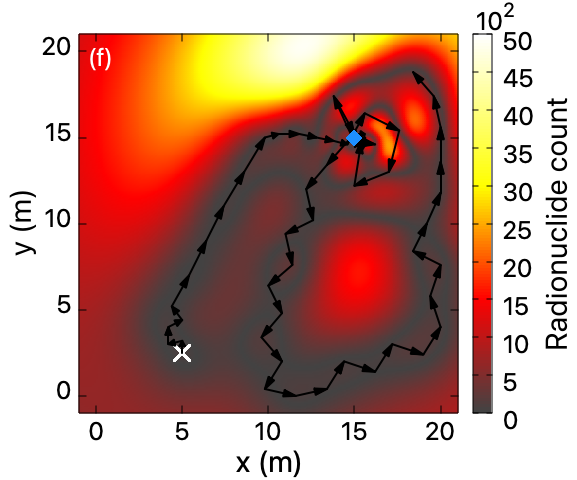}
        \end{subfigure}
        \caption*{Completed flight.}
    \end{subfigure}

    \caption{Time snapshots of the GP mean and error along the flight trajectory (arrows) generated using the GP-DUCB algorithm in (\ref{eq:gp_ducb_algorithm}) with $\rho_t = 0.0313$ for $t \in \{1, ..., T\}$ and $T$ as per (\ref{eq:stopping_condition}). The x-mark and diamond indicate the initial platform and source positions, respectively.}
    \label{fig:scenario_1}
\end{figure}

Figure~\ref{fig:scenario_1} presents time snapshots of the flight trajectory over the $^{137}$Cs source environment using the GP-DUCB algorithm with a constant switching cost of $\rho_t=0.0313$ (soft exploration radius $\bar{r}=4$ m). The left column illustrates the GP mean estimate of the radiation field, while the right column depicts the absolute error with respect to the true radiation field, determined by (\ref{eq:measurement_model}) without background radiation. At 2 minutes into operation, the GP mean exhibits a spurious maximum, indicated by high GP error in Figure~\ref{fig:scenario_1}b. By 6 minutes, the global maximum is localised and the algorithm starts exploring regions of higher uncertainty (see Figure~\ref{fig:scenario_1}c). This leads to the completed flight trajectory, where measurement collection in uncertain areas improves the accuracy of the radiation field estimate, as indicated in Figure~\ref{fig:scenario_1}f.

\begin{figure}[h!]
    \centering
    % --- Row 1 ---
    \begin{subfigure}{0.49\linewidth}
        \centering
        \includegraphics[width=\linewidth]{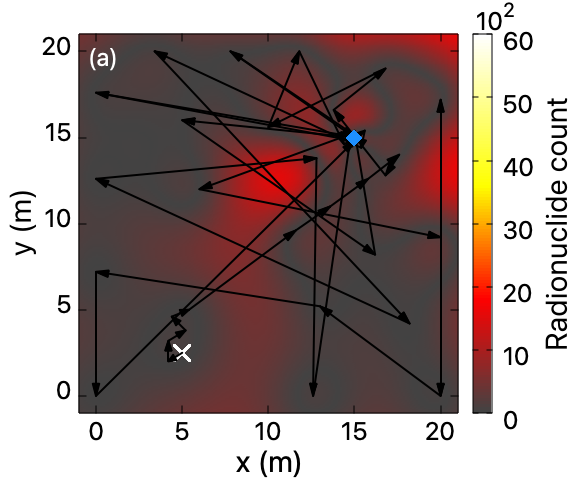}
        \caption*{$\rho_t = 0$.}
        \label{fig:scenario1_rho0}
    \end{subfigure}
    \hfill
    \begin{subfigure}{0.49\linewidth}
        \centering
        \includegraphics[width=\linewidth]{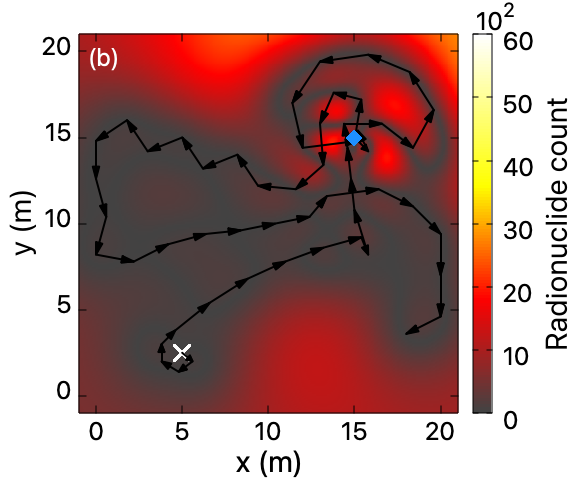}
        \caption*{$\rho_t = 0.02$.}
        \label{fig:scenario1_rho002}
    \end{subfigure}

    \vspace{0em}

    % --- Row 2 ---
    \begin{subfigure}{0.49\linewidth}
        \centering
        \includegraphics[width=\linewidth]{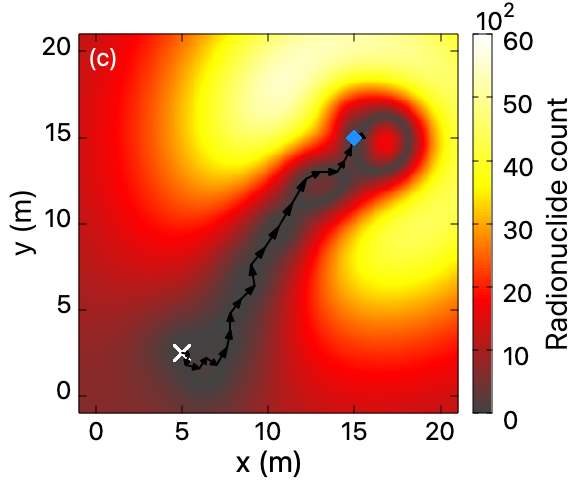}
        \caption*{$\rho_t = 0.125$.}
        \label{fig:scenario1_rho05}
    \end{subfigure}
    \hfill
    \begin{subfigure}{0.49\linewidth}
        \centering
        \includegraphics[width=\linewidth]{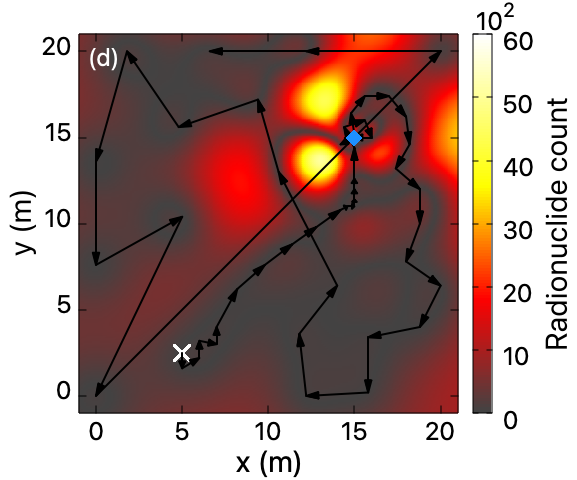}
        \caption*{$\rho_t = 0.125\exp(-\frac{(t-20)^2}{100})$.}
        \label{fig:scenario1_rho_var}
    \end{subfigure}
    \caption{GP error along flight trajectories (arrows) generated by the GP-DUCB algorithm in (\ref{eq:gp_ducb_algorithm}) with different $\rho_t$ values for $t \in \{1, ..., T\}$ and $T$ as per (\ref{eq:stopping_condition}). The x-mark and diamond indicate the initial platform and source positions, respectively.}
    \label{fig:scenario_1_time_varying}
\end{figure}

\begin{figure}[h!]
\centering
\includegraphics[width=0.9\linewidth]{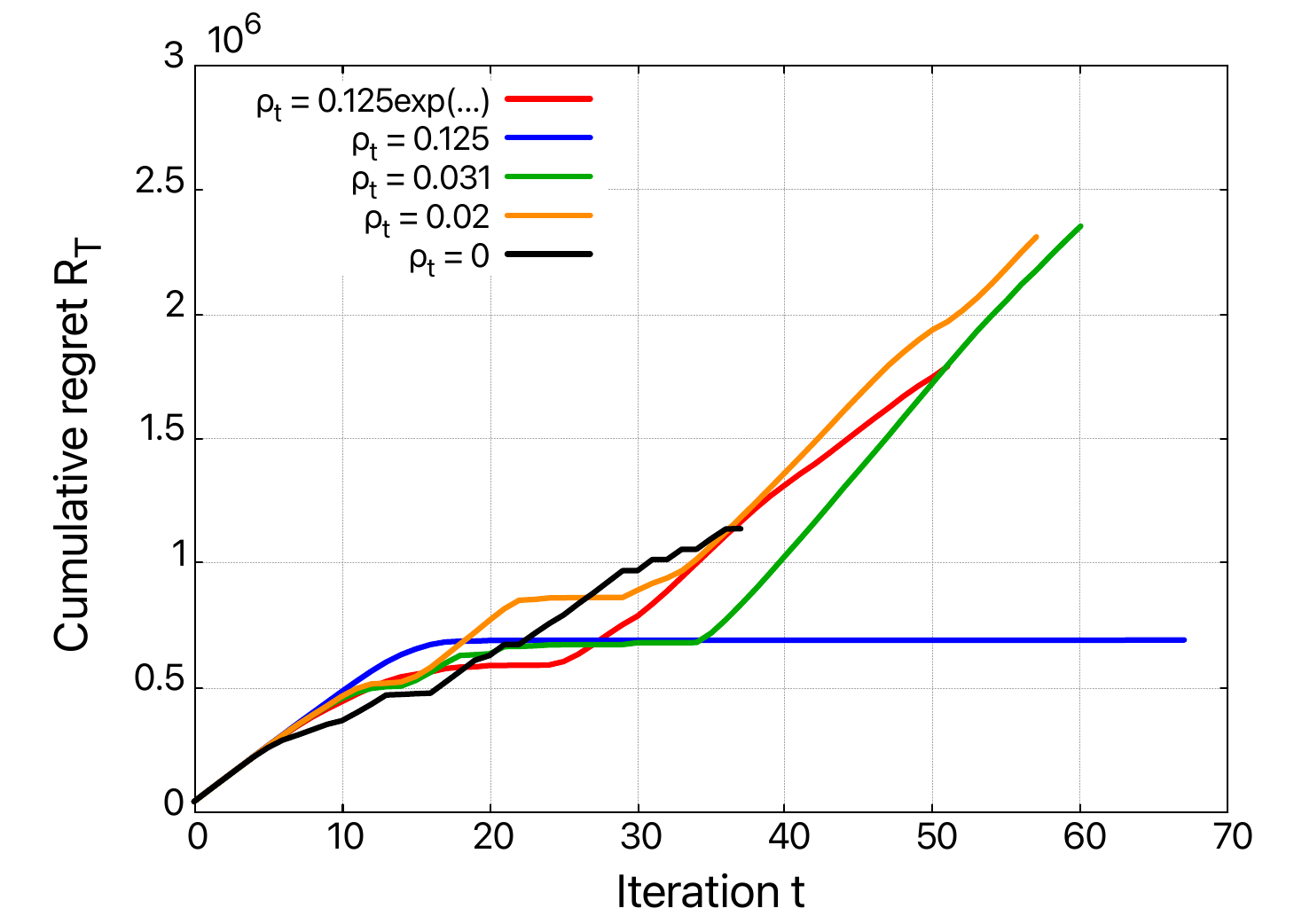}
\caption{Cumulative regret as per (\ref{cummul_reg}) for the flight trajectories generated by the GP-DUCB algorithm with varying $\rho_t$ for $t \in \{1, ..., T\}$ with $T$ defined in (\ref{eq:stopping_condition}).}
\label{fig:regret_scenario_1}
\end{figure}
\unskip

Figure~\ref{fig:scenario_1_time_varying} illustrates the impact of varying the switching cost $\rho_t$ in the GP-DUCB algorithm on the flight trajectory. With $\rho_t=0$ (see Figure~\ref{fig:scenario_1_time_varying}a), corresponding to the GP-UCB strategy, the platform moves extensively between measurements. Increasing the cost to $\rho_t=0.02$ ($\bar{r}=5$ m, Figure~\ref{fig:scenario_1_time_varying}b) yields a more compact and predictable path. A higher cost $\rho_t=0.125$ ($\bar{r}=2$ m, Figure~\ref{fig:scenario_1_time_varying}c) drives the platform rapidly to the source with limited exploration, increasing the uncertainty around the maximum. The time-varying cost $\rho_t=0.125\exp(-(t-20)^2/100)$ in Figure~\ref{fig:scenario_1_time_varying}d, enforces a soft exploration radius of $\bar{r}=2$~m at $t=20$, promoting source seeking initially before transitioning to broader exploration. This behaviour aligns with the cumulative regret plot in Figure~\ref{fig:regret_scenario_1}. Initially, the regret in the time-varying case increases until it stablises near the function maximum due to increased cost, before resuming exploration and increasing the cumulative regret.

\begin{figure}[h!]
\includegraphics[width=0.495\linewidth]{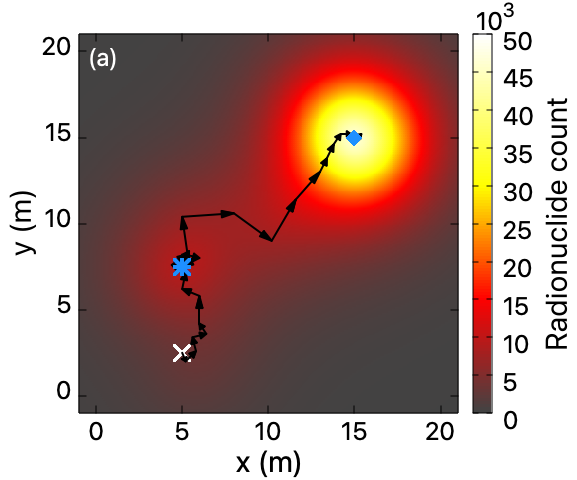}
\includegraphics[width=0.495\linewidth]{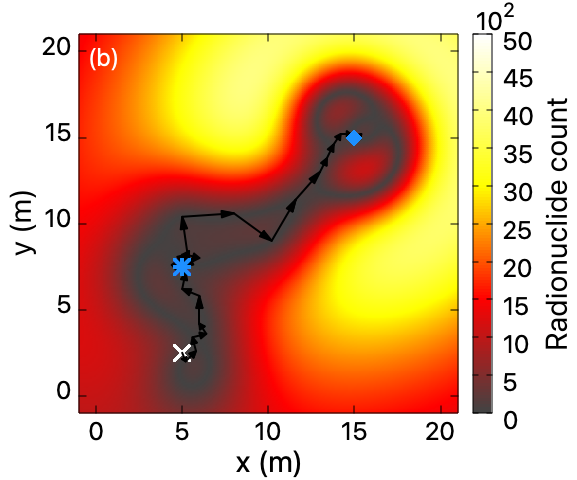}
\caption{GP mean (a) and error (b) along the trajectory (arrows), $T$ as per \eqref{eq:stopping_condition}, generated as per \eqref{eq:gp_ducb_algorithm} with constant $\rho_t = 0.005$. The $^{137}$Cs source positions are indicated by the diamond (primary) and star (secondary) markers, while the x-mark indicates the initial platform position.
}
\label{fig:scenario_2}
\end{figure}
\unskip

A significant benefit of BO is global maximisation. This is demonstrated in Figure~\ref{fig:scenario_2} with the introduction of an additional $^{137}$Cs source (150 MBq) at $q=(5 \text{ m}, 7.5 \text{ m})$. The flight trajectory was generated with $\rho_t=0.005$ ($\bar{r}= 10$ m), and shown to localise the global maximum of the radiation field after exploring the local maximum around the secondary $^{137}$Cs source.

\section{Conclusions}\label{sec:conclusion}
This work proposes a sample-efficient method for radioactive source seeking that utilises Bayesian optimisation with heteroscedastic Gaussian process regression. This method extends the Gaussian Process Upper Confidence Bound (GP-UCB) framework of \cite{Srinivas2012} by incorporating a movement switching cost to discourage excessive inter-sample travel between measurement points, and a heteroscedastic GP to account for the noise characteristics of radiation measurements. This modified strategy, dubbed Gaussian Process Distance Upper Confidence Bound (GP-DUCB), was shown to achieve sublinear regret. Simulation studies demonstrated the performance of the algorithm under both constant and time-varying switching costs. In addition, the algorithm successfully localised the global maximum of the radiation field in a dual-radioactive source environment, despite initially exploring a local maximum. Future work will examine extensions to multi-source localisation and the integration of multi-agent systems for cooperative source seeking. 

\bibliography{export}

\end{document}